\documentclass[
oneside,
reqno,
11pt
]
{amsart}
\usepackage[utf8]{inputenc}
\usepackage{amsmath}
\usepackage{amsfonts}
\usepackage{amssymb}
\usepackage{amsthm}
\usepackage{bbm}
\usepackage[a4paper]{geometry}
\usepackage{color}
\usepackage{latexsym}
\usepackage{bigints}
\usepackage[colorlinks=true,linkcolor=blue,citecolor=mauve]{hyperref}

\definecolor{mauve}{rgb}{0.58,0,0.82}

\def\R{\mathbb R}

\usepackage{esint}

\let\phi\varphi
\let\epsilon\varepsilon
\newtheorem{thm}{Theorem}[section]
\newtheorem{prop}[thm]{Proposition}

\theoremstyle{definition}

\newtheorem*{rem}{Remark}

\numberwithin{equation}{section}
\address{Simon Barth, Institut f\"ur Analysis, Dynamik und Modellierung, Universit\"at Stuttgart, Pfaffenwaldring 57, D-70569 Stuttgart, Germany}
\email{simon.barth@mathematik.uni-stuttgart.de}

\address{Andreas Bitter, Institut f\"ur Analysis, Dynamik und Modellierung, Universit\"at Stuttgart, Pfaffenwaldring 57, D-70569 Stuttgart, Germany}
\email{andreas.bitter@mathematik.uni-stuttgart.de}

\author{Simon Barth and Andreas Bitter}
\date{}
\title[]{\ \textbf{Decay rates of bound states at the spectral threshold of multi-particle Schr\"odinger operators} \ \\ }
\begin{document}
\maketitle 
\begin{abstract} 
We consider $N$-body Schr\"odinger operators with $N\geq3$ particles in dimension $d\geq 3$ in the critical case when the lowest eigenvalue coincides with the bottom of the essential spectrum of the operator. We give the asymptotic behaviour of the corresponding eigenfunction in dependence of the dimension and the number of particles of the system.
\end{abstract}
\pagestyle{headings}
\renewcommand{\sectionmark}[1]{\markboth{#1}{}}
\section{Introduction}
It is well known that if a Schr\"odinger operator has an eigenvalue lying below its essential spectrum, then the corresponding eigenfunctions decay exponentially \cite{agmon}. However, the situation changes completely at the threshold of the essential spectrum. In the work at hand we consider a multi-particle Schr\"odinger operator
\begin{equation}\label{Definition Hamiltonian}
H_N(\lambda) = - \sum\limits_{i=1}^N \frac{1}{2m_i} \Delta_{x_i} + \lambda \sum\limits_{1\le i<j\le N} V_{ij}(x_{ij})
\end{equation}
with a coupling constant $\lambda>0$. After the separation of the center of mass motion this operator can be written as
\begin{equation}
H_0(\lambda)=-\Delta_0+\lambda\sum\limits_{1\le i<j\le N} V_{ij}(x_{ij}),
\end{equation}
where $\Delta_0$ is the Laplace-Beltrami operator on the space of relative motion of the system. We study the case when $\lambda$ takes its critical value, i.e. for $0<\lambda \leq \lambda_0$ the spectrum of the Hamiltonian $H_0(\lambda)$ coincides with the half-line $[0,\infty)$ and for $\lambda>\lambda_0$ sufficiently close to $\lambda_0$ the operator $H_0(\lambda)$ has negative eigenvalues. The decay properties of solutions of the Schr\"odinger equation corresponding to the critical coupling constant $\lambda_0$ play an important role in different physical phenomena, see for example \cite{klaussimon1} and \cite{klaussimon2}. They are in particular crucial for the existence and non-existence of the so-called Efimov effect \cite{cmp}. In a recent work \cite{second} it was shown that in case of short-range potentials for $\lambda=\lambda_0$ the operator $H_0(\lambda)$ has an non-degenerate eigenvalue at zero, where the corresponding eigenfunction $\varphi_0$ satisfies
\begin{equation}
(1+|x|)^\alpha \varphi_0 \in L^2 \qquad \text{for any} \qquad \alpha < \frac{d(N-2)}{2}-2.
\end{equation}
Here $N\geq 3$ is the number of particles and $d\geq 3$ is the dimension. However, the question remains open whether this estimate from below on the decay rate of $\varphi_0$ is close to the optimal one.
\\ In this work we answer this question and give an explicit formula for the asympotitics of the eigenfunction $\varphi_0$ for large distances to the center of mass.

\section{Notation and Main Result}
We consider a system of $N\ge 3$ particles in dimension $d\ge 3$ with masses $m_i>0$ and position vectors $x_i \in \R^d, \ i=1,\dots, N$. Such a system is described by the Hamiltonian \eqref{Definition Hamiltonian}, where $\lambda>0$ is a real parameter, $x_{ij}=x_i-x_j$ and the potentials $V_{ij}$ describe the pair interactions. We assume that $V_{ij}$ satisfy
\begin{equation}\label{Condition potential decay}
	\vert V_{ij}(x_{ij})\vert\le c\vert x_{ij} \vert^{-2-\nu}, \quad x_{ij}\in \mathbb{R}^d, \ \vert x_{ij}\vert \ge A
\end{equation}
for some constants $c,\nu, A>0$ and
\begin{equation}\label{Condition potential local}
\begin{cases}
	V_{ij}\in L_{\mathrm{loc}}^{2}(\mathbb{R}^d),  &\text{if } d=3,\\  V_{ij}\in L_{\mathrm{loc}}^{p}(\mathbb{R}^d) \ \text{for some } p> 2,& \text{if } d= 4,\\ V_{ij}\in L_{\mathrm{loc}}^{\frac{d}{2}}(\mathbb{R}^d), & \text{if } d\ge 5.
	\end{cases}
\end{equation}
Under these assumptions on the potentials the operator $H_N$ is essentially self-adjoint on $L^2(\mathbb{R}^{dN})$. Following \cite{Sigalov}, we define the space $R_0$ of relative motion of the system as
\begin{equation}
	R_0=\left\{x=(x_1,\dots,x_N) \in \mathbb{R}^{dN} \, : \, \sum_{i=1}^N m_ix_i =0\right\}
\end{equation}
and the scalar product
\begin{equation}
	\langle x,\tilde x\rangle_1=\sum_{i=1}^N m_i\langle x_i,\tilde x_i\rangle, \qquad \vert x\vert_1^2 = \langle x,x\rangle_1.
\end{equation}
For a fixed pair of particles $i\neq j$ we set
\begin{equation}\label{eq: Rij}
	R_{ij} = \{x=(x_1,\dots,x_N)\in R_0\,:\, m_ix_i+m_jx_j=0\} \quad \text{and} \quad R_{ij}^\perp =R_0\ominus R_{ij}.
\end{equation}
Let $P_{ij}$ and $P_{ij}^\perp$ be the projections in $R_0$ with respect to the scalar product $\langle\cdot,\cdot\rangle_1$ onto $R_{ij}$ and $R_{ij}^\perp$, respectively. For $x\in R_0$ we denote
\begin{equation}
	q_{ij}= P_{ij} x\qquad \qquad   \text{and} \qquad \qquad \xi_{ij} = P_{ij}^\perp x.	
\end{equation}
Note that for any $1\le i<j\le N$ it holds
\begin{equation}\label{eq: qij and xij}
	|q_{ij}|_1 = \frac{\sqrt{m_im_j}}{\sqrt{m_i+m_j}}{|x_{ij}|},
\end{equation}
which together with \eqref{Condition potential decay} implies
\begin{equation}\label{eq: potential estimate q}
	\vert V_{ij}(x_{ij})\vert\le C |q_{ij}|_1^{-2-\nu} \quad \text{ for some } C>0 \text{ and all } |x_{ij}|\ge A.
\end{equation}
In the following we denote
\begin{equation}
	V(x) = \sum_{1\le i<j\le N} V_{ij}(x_{ij}).
\end{equation}
By $\Delta_0$ we denote the Laplace Beltrami operator on $L^2(R_0)$. Then the Hamiltonian of the system after separation of the center of mass is given by
\begin{equation}
 H_0(\lambda)=-\Delta_0+ \lambda V(x).
\end{equation}
In this work we consider systems for which the following important condition is fulfilled.
\\ \textbf{Assumption:}
We consider the case when $\lambda$ takes its critical value $\lambda_0$, i.e. for $\lambda \le \lambda_0$ the spectrum of $H_0(\lambda)$ is $[0,\infty)$ and for $\lambda>\lambda_0$ the operator $H_0(\lambda)$ has negative spectrum. We assume that for all $\lambda>\lambda_0$ sufficiently close to $\lambda_0$ the essential spectrum of the operator $H_0(\lambda)$ coincides with the half line $[0,\infty)$. For such $\lambda$ the negative spectrum is purely discrete. Without loss of generality we set $\lambda_0=1$ and we write $H_0$ instead of $H_0(\lambda_0)$. In \cite{second} it was shown that under these assumptions in case of $d\ge 3$ and $N\ge 3$ zero is a non-degenerate eigenvalue of $H_0$. The following theorem gives the asymptotic behavior of the corresponding eigenfunction for large arguments. 
\begin{thm}\label{thm: decay of virtual levels}
Assume that $H_0$ satisfies the conditions described above. Suppose that $\varphi_0$ is an eigenfunction of $H_0$ corresponding to the eigenvalue zero. Then the following assertions hold. 
\begin{enumerate}
 	\item[\textbf{(i)}] For all $1\le i<j\le N$ we have
 	\begin{equation}\label{eq: 1+x gamma}
 	V_{ij}(x_{ij})\varphi_0(x)\in L^1(R_0).
 	\end{equation}
 	\item[\textbf{(ii)}] Let $\beta=d(N-1)-2$ and denote by $|\mathbb{S}^{\beta-1}|$ the volume of the unit sphere in $\mathbb{R}^\beta$. Further, let
 	\begin{equation}
C_0=-\frac{1}{(\beta-2)|\mathbb{S}^{\beta-1}|}\int\limits_{R_0} \sum_{1\le i<j\le N}V_{ij}(x_{ij})\varphi_0(x)\,\mathrm{d}x.
	\end{equation}
 	Then the function $\varphi_0$ has the following asymptotics
\begin{equation}\label{eq: expansion phi0}
	\varphi_0(x) = \frac{C_0}{\vert x\vert_1^{\beta}} +g(x) \quad \text{as } \quad \vert x\vert_1 \rightarrow \infty,
	\end{equation} 
where the remainder $g$ belongs to $L^p(R_0)$ for any $p$ satisfying
\begin{equation}\label{eq: p}
\frac{\beta+2}{\beta+\frac{\gamma^\ast}{1+\gamma^\ast}}<p<\frac{\beta+2}{\beta} \qquad \text{with} \quad \gamma^{\ast}=\min\left\{\frac{d}{2}-1,\nu \right\}.
\end{equation}	
\end{enumerate}
\end{thm}

\begin{rem}
\begin{enumerate}
\item[\textbf{(i)}] Note that $\varphi_0$ can be chosen to be strictly positive. If for all $V_{ij}$ it holds $V_{ij}(x)\leq 0$, then we have $C_0\not =0$. In this case the leading term $C_0|x|_1^{-\beta} \chi_{\{|x|_1 > 0\}}$ belongs to $L^q(R_0)$, only if $q> \frac{\beta+2}{\beta}$.
\item[\textbf{(ii)}] \eqref{eq: expansion phi0} shows that the decay rate of $\varphi_0$ does not depend on the potentials as long as the pair potentials are short-range. At the same time, since $|x|_1=\sum_{i=1}^N m_ix_i$, the decay of $\varphi_0$ depends on the direction, if the masses are not equal.
\end{enumerate}
\end{rem}

\begin{proof}[Proof of Theorem \ref{thm: decay of virtual levels}]
We will split the proof of the theorem into several propositions. The statement of the following proposition was proved in \cite{second}.
\begin{prop}\label{prop from second}
The function $\varphi_0$ satisfies 
	\begin{equation}
		\nabla_0\left(|x|_1^\alpha\varphi_0\right)\in L^2(R_0) \qquad \ \ \text{ for any }\quad 0\le \alpha< \frac{d(N-1)-2}{2}.
	\end{equation}
\end{prop}
The statement of assertion (i) of Theorem \ref{thm: decay of virtual levels} is a special case of the following
\begin{prop}\label{Prop 1}
For all $1\le i<j\le N$ and any $0<{\gamma}<\gamma^\ast$ we have
 	\begin{equation}
 	(1+|x|_1)^{\gamma} V_{ij}(x_{ij})\varphi_0(x)\in L^1(R_0).
 	\end{equation}
\end{prop}
\begin{proof}[Proof of Proposition \ref{Prop 1}]
By Proposition \ref{prop from second}, together with $|\nabla_{q_{ij}}| \le |\nabla_0|$ and Hardy's inequality in the space $H^1(R_{ij})$ we have
\begin{equation}\label{eq: Hardy in q}
 \left(1+|q_{ij}|_1\right)^{-1}\left(1+|x|_1\right)^\alpha\varphi_0 \in L^2(R_0).
\end{equation}
Note that the potential $V_{ij}$ decays in the direction $x_{ij}$ but not necessarily in all directions. We will combine the decay property \eqref{Condition potential decay} of $V_{ij}$ and the a priori estimate \eqref{eq: Hardy in q} of $\varphi_0$ to get \eqref{eq: 1+x gamma}.
\\
For any fixed $0<{\gamma} <\gamma^\ast$ we write
\begin{equation}
\left(1+|x|_1\right)^{\gamma} V_{ij}(x_{ij}) \varphi_0(x) = \left(1+|q_{ij}|_1\right)^{-1}\left(1+|x|_1\right)^\alpha\varphi_0(x)\cdot f(x),
\end{equation}
where
\begin{equation}\label{eq: function f}
f(x):=\left(1+|x|_1\right)^{-\alpha+{\gamma}}\left(1+|q_{ij}|_1\right)V_{ij}(x_{ij}).
\end{equation}
In view of \eqref{eq: Hardy in q} to prove Proposition \ref{Prop 1} it suffices to show that $f$ belongs to $L^2(R_0)$. We decompose the function $f$ as
\begin{equation}
	f(x)=f(x)\chi_{\{|x_{ij}| < A\}}+f(x)\chi_{\{|x_{ij}| \ge A\}}
\end{equation}
and estimate the functions $f(x)\chi_{\{|x_{ij}| < A\}}$ and $f(x)\chi_{\{|x_{ij}| \ge A\}}$ separately, starting with the function $f(x)\chi_{\{|x_{ij}| < A\}}$. Note that $L^2(R_0) = L^2(R_{ij})\otimes L^2(R_{ij}^\perp).$\\
Due to \eqref{eq: qij and xij} and assumption \eqref{Condition potential local} it holds $(1+|q_{ij}|_1)V_{ij}(x_{ij})\chi_{\{|x_{ij}|< A\}}\in L^2(R_{ij})$. Moreover, we can estimate $(1+|x|_1)^{-1}\le (1+|\xi_{ij}|_1)^{-1}$. Since $\dim(R_{ij}^\perp) = d(N-2)$, we have
\begin{equation}\label{eq: estimate 1+xi for q small}
	 \left(1+|\xi_{ij}|_1\right)^{-\alpha+{\gamma}} \in L^2(R_{ij}^\perp) \quad \text{if and only if} \quad \alpha-{\gamma} > \frac{d(N-2)}{2}.
\end{equation}
Recall that ${\gamma}<\gamma^\ast$, which in particular implies that $\gamma<\frac{d}{2}-1$. Therefore, the condition in \eqref{eq: estimate 1+xi for q small} is satisfied if we choose $\alpha$ close enough to $\frac{d(N-1)-2}{2}$. Hence, we have $\left(1+|x|_1\right)^{-\alpha+{\gamma}} \in L^2(R_{ij}^\perp)$
and therefore
\begin{equation}
f(x) \chi_{\{|x_{ij}| < A\}}\in L^2(R_0).
\end{equation}
In order to prove that the function $f(x) \chi_{\{|x_{ij}| \ge A\}}$ belongs to $L^2(R_0)=L^2(R_{ij})\otimes L^2(R_{ij}^\perp)$, we
show that it can be estimated as
\begin{equation}
	|f(x) \chi_{\{|x_{ij}| \ge A\}}|\le |f_1(q_{ij})|\cdot |f_2(\xi_{ij})|,
\end{equation}
where $f_1\in L^2(R_{ij})$ and $f_2\in L^2(R_{ij}^\perp)$. Here, we will use the assumption that the potential $V_{ij}(x_{ij})$ decays faster than $|q_{ij}|_1^{-2}$ as $|x_{ij}|\rightarrow\infty$. Recall that $\dim(R_{ij}) =d$ and $\dim(R_{ij}^\perp) = d(N-2)$, which implies that for any $0<\varepsilon< \nu -\gamma$ we have
\begin{equation}\label{eq: f1}
f_1(q_{ij}):=\left(1+|q_{ij}|_1\right)^{-\frac{d}{2}-\varepsilon}\in L^2(R_{ij})  
\end{equation}
and 
\begin{equation}\label{eq: f2}
\qquad \qquad\; f_2(\xi_{ij}) := \left(1+|\xi_{ij}|_1\right)^{-\alpha+{\gamma}-\nu+\varepsilon+\frac{d}{2}-1}\in L^2(R_{ij}^\perp).
\end{equation}
Note that we can always assume $\nu< \frac{d}{2}-1$. By the use of $|q_{ij}|_1, |\xi_{ij}|_1\le |x|_1$ we get
\begin{equation}\label{eq: inequality x q xi}
	\left(1+|x|_1\right)^{-\alpha+{\gamma}}\le \left(1+|\xi_{ij}|_1\right)^{-\alpha+{\gamma}-\nu+\varepsilon+\frac{d}{2}-1}\left(1+|q_{ij}|_1\right)^{1-\frac{d}{2}+\nu-\varepsilon}.
\end{equation}
This, together with \eqref{eq: potential estimate q} yields
\begin{equation}
	|f(x) \chi_{\{|x_{ij}| \ge A\}}|\le C |f_1(q_{ij})|\cdot |f_2(\xi_{ij})|
\end{equation}
and therefore $f(x) \chi_{\{|x_{ij}| \ge A\}}\in L^2(R_0)$, which completes the proof of Proposition \ref{Prop 1}.
\end{proof}
Now we turn to the proof of assertion (ii) of Theorem \ref{thm: decay of virtual levels}. Since
\begin{equation}
H_0\varphi_0 = \left(-\Delta_0 + V\right)\varphi_0 =0
\end{equation}
and due to Proposition \ref{Prop 1} $V\varphi_0\in L^1(R_0)$, we can apply Theorem 6.21 in \cite{Loss} to conclude
\begin{equation}\label{eq: Green Formel}
	\varphi_0(x)  = \frac{-1}{(\beta-2)\vert  \mathbb{S}^{\beta-1}\vert}\bigintsss_{R_0} {\vert x-y\vert_1^{-\beta}}{V(y)\varphi_0(y)}\; \mathrm{d}y.
\end{equation}
We derive the asymptotics \eqref{eq: expansion phi0} by studying the integral representation of $\varphi_0$ in \eqref{eq: Green Formel}. We will see that only certain regions contribute to the leading term in \eqref{eq: expansion phi0}. We write
\begin{equation}
	\varphi_0(x)=\frac{-1}{(\beta-2)\vert  \mathbb{S}^{\beta-1}\vert}\left(I_1(x)+I_2(x)\right), 
\end{equation}
where
\begin{align}
\begin{split}
	I_1(x) &= \int\limits_{\{|x-y|_1\le 1\}} {\vert x-y\vert_1^{-\beta}}{V(y)\varphi_0(y)}\; \mathrm{d}y, \\
	I_2(x) &= \int\limits_{\{|x-y|_1>1 \}} {\vert x-y\vert_1^{-\beta}}{V(y)\varphi_0(y)}\; \mathrm{d}y. \\
	\end{split}
\end{align}
At first, we prove that the function $I_1$ belongs to the remainder $g$ in \eqref{eq: expansion phi0}, as we can see in the following
\begin{prop}\label{Prop 2}
The function $I_1$ is an element of $L^p(R_0)$ {for all} $1\le p<\frac{\beta+2}{\beta}$.
\end{prop}
\begin{proof}[Proof of Proposition \ref{Prop 2}]
Due to $\dim(R_0)=d(N-1)$ and $\beta=d(N-1)-2$ we have
\begin{equation}
|x|_1^{-\beta}\chi_{\{|x|_1\le 1\}}\in L^p(R_0) \quad \text{for all}\quad 1\le p<\frac{\beta+2}{\beta}.
\end{equation}
By Proposition \ref{Prop 1} we have $V\varphi_0 \in L^1(R_0)$, which together with Young's inequality yields the claim of Proposition \ref{Prop 2}.
\end{proof}
Now we will show that only a part of $I_2$ gives the leading term in \eqref{eq: expansion phi0}. Let $\eta = \frac{1}{1+\gamma^\ast}$. For $x\in R_0$ we define
\begin{align}\label{eq: domains Omega }
\begin{split}
	 \Omega_1(x)&= \left\{y\in R_0\;:\; \vert x-y\vert_1 >1, \ \vert y\vert_1 > \vert x\vert_1^{\eta} \right\},\\
	\Omega_2(x)&=\left\{y\in R_0\;:\; \vert x-y\vert_1 >1, \ \vert y\vert_1 \le  \vert x\vert_1^{\eta} \right\}
\end{split}
\end{align}
and 
\begin{align}\label{eq: functions I21}
\begin{split}
	I_{2,k}(x) &= \bigintsss_{\Omega_k(x)}{\vert x-y\vert_1^{-\beta}}{V(y)\varphi_0(y)}\; \mathrm{d}y, \qquad k=1,2.
	\end{split}
\end{align}
We will show that only the function $I_{2,2}$ contributes to the leading term in \eqref{eq: expansion phi0}. At first we consider the function $I_{2,1}$ and show that it belongs to the remainder in \eqref{eq: expansion phi0}. Indeed, we have the following
\begin{prop}\label{Prop 3}
Let $I_{2,1}$ be given by \eqref{eq: domains Omega } and \eqref{eq: functions I21}. We have
\begin{equation}\label{eq: Prop 3}
	I_{2,1}\in L^p(R_0)\quad\text{for all}\quad  \frac{\beta+2}{\beta+\frac{\gamma^\ast}{1+\gamma^\ast}}<p<\frac{\beta+2}{\beta}.
\end{equation}
\end{prop}
\begin{proof}[Proof of Proposition \ref{Prop 3}]
In the proof we will use Proposition \ref{Prop 1}. Let $\gamma<\gamma^\ast$. Using $\vert y\vert_1 > \vert x\vert_1^{\eta}$ for $y\in \Omega_1(x)$ we get
\begin{align}\label{eq: estimate I21}
\begin{split}
|I_{2,1}(x)| &\le \bigintsss_{\Omega_1(x)}{\vert x-y\vert_1^{-\beta}}{|V(y)\varphi_0(y)|}\; \mathrm{d}y\\
& \le \left(1+|x|_1^{\eta}\right)^{-{\gamma}}\bigintsss_{\Omega_1(x)}{\vert x-y\vert_1^{-\beta}}{\left(1+|y|_1\right)^{\gamma}|V(y)\varphi_0(y)|}\; \mathrm{d}y.
\end{split}
\end{align}
We show that for any fixed $p$ satisfying \eqref{eq: Prop 3} we find a constant $\gamma<\gamma^\ast$, such that the function on the r.h.s. of \eqref{eq: estimate I21} belongs to $L^p(R_0)$. Note that $\frac{\gamma^\ast}{1+\gamma^\ast}=\eta \gamma^\ast$ $\gamma $ sufficiently close to $\gamma^\ast$ it holds
\begin{equation}\label{eq: condition p}
	p > \frac{\beta+2}{\beta+\eta{\gamma}}.
\end{equation}
By Proposition \ref{Prop 1} and Young's inequality we have
\begin{equation}
\!\!	\tilde{I}_{2,1}(x):= \!\bigintsss_{\Omega_1(x)}{\! \!\! \!\vert x-y\vert_1^{-\beta}}{\left(1+|y|_1\right)^{\gamma}\!|V(y)\varphi_0(y)|}\; \mathrm{d}y \in L^s(R_0), \ \  s>\frac{d(N-1)}{d(N-1)-2}.
\end{equation}
Now we apply H\"older's inequality to the r.h.s. of \eqref{eq: estimate I21}. For this purpose, we fix a constant $s>\frac{d(N-1)}{d(N-1)-2}$ and define
\begin{equation}
t_1 = \frac{s}{s-p}\ge 1 \quad \text{and} \quad t_2= \frac{s}{p}\ge 1 \quad \text{with}\quad \frac{1}{t_1}+\frac{1}{t_2}=1.
\end{equation}
Then we formally get
\begin{equation}\label{eq: Hoelder I21 formally}
	\!\!\!\int_{R_0}\!\!\left(1+\vert x\vert_1^{\eta}\right)^{\!-{\gamma} p} \!\vert \tilde{I}_{2,1}(x)\vert^p\, \mathrm{d}x\le \!\left(\int_{R_0}\!\! \left(1+\vert x\vert_1^{\eta}\right)^{-{\gamma} p t_1}\! \mathrm{d}x\right)^{\frac{1}{t_1}}\!\!\left(\int_{R_0}\! |\tilde{I}_{2,1}(x)|^{pt_2}\, \mathrm{d}x\right)^{\frac{1}{t_2}}\,.
\end{equation}
Since $pt_2=s$ and $\tilde{I}_{2,1}\in L^s(R_0)$, the second integral on the r.h.s of \eqref{eq: Hoelder I21 formally} is finite. Due to $\dim(R_0)=d(N-1)$, to prove the finiteness of the first integral on the r.h.s of \eqref{eq: Hoelder I21 formally} it suffices to show that ${\eta \gamma} pt_1>d(N-1)$. By definition of $t_1$ this is equivalent to
\begin{align}\label{eq: Hoelder I21}
\begin{split}
\eta {\gamma} sp > d(N-1)(s-p)\quad &\Leftrightarrow\quad  p(\eta {\gamma} s+d(N-1))>d(N-1)s \\ &\Leftrightarrow\quad \frac{1}{p}< \frac{\eta s{{\gamma}}+d(N-1)}{sd(N-1)}= \frac{\eta{\gamma}}{d(N-1)}+\frac{1}{s}.
\end{split}
\end{align}
Since $p>\frac{d(N-1)}{d(N-1)-2+ {\gamma} \eta} $, we see that the condition in \eqref{eq: Hoelder I21} is fulfilled if $s$ is chosen sufficiently close to $\frac{d(N-1)}{d(N-1)-2}$. Since $\beta = d(N-1)-2$, this completes the proof of Proposition \ref{Prop 3}.
\end{proof}
Now we finally show that the function $I_{2,2,}$ yields the leading term in \eqref{eq: expansion phi0}.
\begin{prop}\label{Prop 4}
Let $I_{2,2}$ be given by \eqref{eq: functions I21}, then we have
\begin{equation}
	I_{2,2}(x) = \vert x\vert_1^{-\beta} \int_{\Omega_2(x)} V(y)\varphi_0(y)\,\mathrm{d}y + h(x), \quad \text{as} \quad |x|_1\rightarrow\infty,
\end{equation}
where 
\begin{equation}
h\in L^p(R_0) \quad\text{ for all}\quad p>\frac{\beta+2}{\beta+\frac{\gamma^\ast}{1+\gamma^\ast}}.
\end{equation}
\end{prop}
\begin{proof}[Proof of Proposition \ref{Prop 4}]
 For $y\in \Omega_2(x)$ it holds (cf. \cite{R4})
\begin{equation}
	\vert x\vert_1^{-1}\big(1-\vert x\vert_1^{\eta-1}\big) \le \vert x-y\vert_1^{-1} \le \vert x\vert_1^{-1}\big(1+c\vert x\vert_1^{\eta-1}\big)
\end{equation}
for some $c>0$. We apply this inequality to the positive and the negative part of the integrand in the definition of $I_{2,2}$ separately. Let
\begin{equation}
(V\varphi_0)_+(x)=\max\left\{V(x)\varphi_0(x),0\right\}\quad \text{and} \quad (V\varphi_0)_-=-(V\varphi_0-(V\varphi_0)_+),
\end{equation}
then we have
\begin{equation}\label{eq: estimate V pm I}
{\vert x\vert_1^{-\beta}\big(1-\vert x\vert_1^{\eta-1}\big)^{\beta}}\int_{\Omega_2(x)}{\big(V\varphi_0\big)_\pm(y)}\; \mathrm{d}y \le \int_{\Omega_2(x)}\frac{\big(V\varphi_0\big)_\pm(y)}{\vert x-y\vert_1^{\beta}}\; \mathrm{d}y 
\end{equation}
and 
\begin{equation}\label{eq: estimate V pm II}
 \int_{\Omega_2(x)}\frac{\big(V\varphi_0\big)_\pm(y)}{\vert x-y\vert_1^{\beta}}\; \mathrm{d}y \le{\vert x\vert_1^{-\beta}\big(1+c\vert x\vert_1^{\eta-1}\big)^{\beta}}\int_{\Omega_2(x)}{\big(V\varphi_0\big)_\pm(y)}\; \mathrm{d}y.
 \end{equation}
Since $\dim(R_0)=d(N-1)$ we see from \eqref{eq: estimate V pm I} and \eqref{eq: estimate V pm II} that there exist functions
\begin{equation}\label{eq: h pm}
h_\pm \in L^p(R_0),\qquad  p>\frac{d(N-1)}{d(N-1)-2+1-\eta},
\end{equation}
such that 
\begin{equation}
\int\limits_{\Omega_2(x)}\frac{\big(V\varphi_0\big)_\pm(y)}{\vert x-y\vert_1^{\beta}}\; \mathrm{d}y = \vert x\vert_1^{-\beta} \int\limits_{\Omega_2(x)} \left(V\varphi_0(y)\right)_\pm\,\mathrm{d}y + h_\pm(x)
\end{equation}
for large $|x|_1$. Hence, we obtain
\begin{equation}\label{eq: I22}
	I_{2,2}(x) =\vert x\vert_1^{-\beta} \int\limits_{\Omega_2(x)} V(y)\varphi_0(y)\,\mathrm{d}y + h(x) \quad \text{ as }\quad |x|_1\rightarrow\infty,
\end{equation}
where $h=h_+-h_-$ belongs to  $L^p(R_0)$ for $p$ given in \eqref{eq: h pm}. Due to $1-\eta = \frac{\gamma^\ast}{1+\gamma^\ast}$ and $\beta=d(N-1)-2$ this concludes the proof of Proposition \ref{Prop 4}.
\end{proof}
By combining Propositions \ref{Prop 2}, \ref{Prop 3} and \ref{Prop 4} we get 
\begin{equation}\label{eq: phi0 mit omega2}
\varphi_0(x)=\frac{-|x|_1^{-\beta}}{(\beta-2)\vert  \mathbb{S}^{\beta-1}\vert}\bigintsss\limits_{\Omega_2(x)} {V(y)\varphi_0(y)}\; \mathrm{d}y+ g(x) \quad \text{ as }\quad |x|_1\rightarrow\infty
\end{equation}
with
\begin{equation}
g\in L^p(R_0) \qquad \text{for}\qquad  \frac{\beta+2}{\beta+ \frac{\gamma}{1+\gamma}} < p <\frac{\beta+2}{\beta}.
\end{equation}
Note that the integral on the r.h.s of \eqref{eq: phi0 mit omega2} is over the set $\Omega_2(x)$, in contrast to \eqref{eq: expansion phi0}, where the integral is over the whole space $R_0$. Therefore, to complete the proof of the theorem it remains to show that
\begin{equation}
|x|_1^{-\beta}\int\limits_{R_0\setminus \Omega_2(x)} V(y)\varphi_0(y)\,\mathrm{d}y 
\end{equation}
does not contribute to the leading term in the asymptotic estimate of $\varphi_0$. Due to Proposition \ref{Prop 1} it is easy to see that for any $\gamma<\gamma^\ast$ we have
\begin{equation}
	\Big\vert\int\limits_{R_0\setminus \Omega_2(x)} V(y)\varphi_0(y)\,\mathrm{d}y \Big\vert \le C \left(1+|x|_1\right)^{-\eta{{\gamma}}} 
\end{equation}
for $|x|_1 $ sufficiently large. This implies 
\begin{equation}\label{eq: small arguments in Lp}
|x|_1^{-\beta}\int\limits_{R_0\setminus \Omega_2(x)} V(y)\varphi_0(y)\,\mathrm{d}y \in L^p(R_0) \quad  \text{for} \quad p>\frac{\beta+2}{\beta+\frac{\gamma}{1+\gamma}}.
\end{equation}
Choosing $\gamma<\gamma^\ast$ sufficiently close to $\gamma^\ast$ and combining \eqref{eq: phi0 mit omega2} and \eqref{eq: small arguments in Lp} completes the proof of the theorem.
\end{proof}

\section*{Acknowledgements}
Both authors are deeply grateful to Timo Weidl and Semjon Vugalter for proposing the problem, useful suggestions and helpful discussions.
\\The main part of the research for this paper was done while visiting the Mittag-Leffler Institute within the semester program \textit{Spectral Methods in Mathematical Physics}. 
\bibliography{Bib}{}
\bibliographystyle{abbrv}
\end{document}